\documentclass[12pt]{iopart}

\usepackage{epsfig}
\usepackage{iopams}

\usepackage{graphicx}
\usepackage{amsthm}
\usepackage{bm}
\usepackage{layout}
\usepackage{float}

\usepackage{amsfonts}
\usepackage{amssymb}%
\begin{document}

\newcommand{\ket}[1]{\left | #1 \right\rangle}
\newcommand{\bra}[1]{\left \langle #1 \right |}
\newcommand{\half}{\mbox{$\textstyle \frac{1}{2}$}}
\newcommand{\smallfrac}[2][1]{\mbox{$\textstyle \frac{#1}{#2}$}}
\newcommand{\braket}[2]{\left\langle #1|#2\right\rangle}
\newcommand{\proj}[1]{\ket{#1}\bra{#1}}
\renewcommand{\epsilon}{\varepsilon}

\newcommand{\cE}{\mathcal{E}}
\newcommand{\cD}{\mathcal{D}}
\newcommand{\cM}{\mathcal{M}}
\newcommand{\ie}{{\it i.e.}}

\newtheorem{remark}{Remark}
\newtheorem{lemma}{Lemma}
\newtheorem{theorem}{Theorem}
\newtheorem{example}{Example}
\newtheorem{definition}{Definition}
\newtheorem{proposition}{Proposition}
\newtheorem{corollary}{Corollary}

\title[Separable states improve protocols with finite randomness]{Separable states improve protocols with finite randomness}

\author{Tan Kok Chuan Bobby$^{1}$
and Tomasz Paterek$^{1,2}$}

\address{
$^1$ Centre for Quantum Technologies, National University of Singapore, Singapore \\
$^2$ School of Physical and Mathematical Sciences, Nanyang Technological University, Singapore}

\date{\today}

\begin{abstract}
It is known from Bell's theorem that quantum predictions for some entangled states cannot be mimicked using local hidden variable (LHV) models.
From a computer science perspective,  LHV models may be interpreted as classical computers operating on a potentially infinite number of correlated bits originating from a common source.
As such, Bell inequality violations achieved through entangled states are able to characterise the quantum advantage of certain tasks, so long as the task itself imposes no restriction on the availability of correlated bits.
However, if the number of shared bits is limited, additional constraints are placed on the possible LHV models and separable, i.e. disentangled, states may become a useful resource. Bell violations are therefore no longer necessary to achieve a quantum advantage.
Here we show that in particular, separable states may improve the so-called random access codes, which is a class of communication problems where one party tries to read a portion of the data held by another distant party in presence of finite shared randomness and limited classical communication.
We also show how the bias of classical bits can be used to avoid wrong answers in order to achieve the optimal classical protocol and how the advantage of quantum protocols is linked to quantum discord.
\end{abstract}

\maketitle

\section{Introduction}

Quantum communication typically study the efficiency of tasks in which either quantum bits are communicated between distant parties, or classical bits are communicated but the parties involved share some quantum correlations.
Many problems can be efficiently solved in this setting and examples include cryptography~\cite{PhysRevLett.67.661}, communication complexity~\cite{buhrman2010}, or computation~\cite{PhysRevLett.102.050502}.
Instances of these problems draw their superiority from the violation of Bell inequalities which require entanglement and the pay off is that such states outperform all classical-like solutions characterised by local hidden variable models.
In the language of computer science, local hidden variable models are models of computation where local classical computers execute algorithms based on input from a potentially unlimited source of random bits.
Here we point out new classes of states, correlated in a quantum way but not necessarily entangled, that may improve quantum protocols if the randomness shared between distant parties consists of a \emph{finite} number of classical or quantum bits.
We prove this rigorously for the task called random access code~\cite{wiesner1983,ambainis1999,ambainis2002} assisted by two bits of randomness. However, due to the general nature of the argument, we expect that a similar reasoning will apply to problems of a similar nature (i.e. where shared randomness is an expensive resource).
Such a restriction is not a limitation of our computing model, as even the universe does not have access to an unlimited number of bits~\cite{PhysRevLett.88.237901}.
Deriving limits on classical computation and communication that take finite randomness into account is therefore not only of practical interest, but may also shed light on fundamental questions.

The present work also contributes a new operational meaning to certain measures of non-classical correlations.
Many quantum states that are not entangled, so called separable states, still posses non-classical features such as those characterised by quantum discord~\cite{henderson2001,ollivier2001,celeri2011,modi2012}.
The role of quantum discord in communication problems was quite extensively studied and connections were established with entanglement transformations~\cite{cavalcanti2011,madhok2011,streltsov2011,piani2011,streltsov2012,chuan2012,modi2012},
coherence of protocols~\cite{madhok2013}, as well as with the performance of certain problems that can be directly compared to their classical counterparts~\cite{boixo2011,dakic2012,horodecki2013}.
However, the latter link with the discord is established only for classical-quantum states~\cite{boixo2011} or for problems with additional constraints such as the lack of certain reference frames~\cite{dakic2012,horodecki2013}.
It is therefore desirable to identify a well-known communication problem with many applications, that can gain efficiencies from discorded states.

In this context, studying random access codes assisted with finite randomness is a natural choice.
Indeed, a quantum version of this problem is as old as quantum information itself~\cite{wiesner1983,ambainis1999,ambainis2002}. The quantum codes were studied in general probabilistic theories~\cite{versteeg2009}, in relation with Popescu-Rohrlich boxes~\cite{grudka2013}, led to information causality~\cite{ic}, find applications in quantum finite automata~\cite{ambainis2002}, quantum communication complexity~\cite{klauck2001}, network coding~\cite{hayashi2006}, security of quantum-key distribution~\cite{li2012}, and have been recently demonstrated experimentally~\cite{spekkens2009}.
Assuming restrictions on shared randomness, we will show that not only do separable, discorded states allow better performance than the best classical solution, they also outperform some entangled states.

%%%%%%%%%%%%%%%%%%%%%%%%%%%%%%%%%%%%%%%%%%%%%%%%%%%%%%%%%%%%%%%%%%%%%%%

\subsection{Random access codes}

Imagine that Bob would like to know (better than just by sheer guess) a random number from Alice's telephone book.
Is it necessary for Alice to send Bob the whole book? Or perhaps she can communicate a fewer number of ``encoded'' pages such that Bob is reasonably confident of getting the correct number?
Random access codes are strategies designed to solve this problem.
As illustrated in Fig.~\ref{FIG_RAC}, in a classical $n \rightarrow 1$ random access code (RAC) Alice receives a random $n$-bit input $x$, communicates a single bit $c$ to Bob, who given this piece of information tries to guess the $i$th bit of Alice, $x_i$, by outputting his guess $b_i$ (in every run $i$ is chosen at random).
 One may construct quantum versions of this task by either having Alice communicate a single quantum bit, or by having Alice and Bob share an entangled quantum state aided by a single bit of classical communication~\cite{pawlowski2010}.
We study here the latter version of the problem and allow for arbitrary quantum states in place of just entangled ones.
The role of quantum discord in the former version of the problem was considered in Ref.~\cite{yao2012}.
Our choice makes the relevance of shared randomness more transparent as by restricting the communication to classical the only additional resource facilitating the process are the assisting (qu)bits.

\begin{figure}
\qquad \qquad \qquad \includegraphics[scale=0.5]{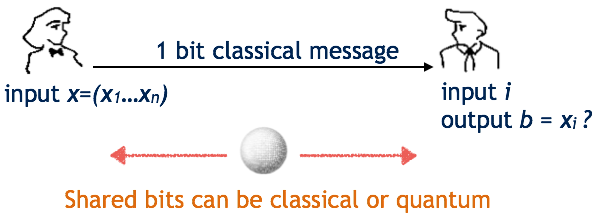}
\caption{$n \to 1$ random access codes with shared randomness.
Alice and Bob share a finite number of classical or quantum bits from a common source (shared randomness).
Alice is allowed to send a single classical bit to Bob, who tries to guess the $i$th bit, $x_i$, of Alice's input string.
We show that sharing quantum bits in separable states improves the worst-case probability of Bob's correct guess over the best classical protocol.}
\label{FIG_RAC}
\end{figure}

The existing quantum codes use a finite number of qubits and are effectively compared with classical protocols with unlimited shared randomness~\cite{ambainis2008,pawlowski2010}.
Under such comparison, the quantum code can outperform the classical ones only if it is assisted by quantum states violating some Bell inequality.
This is because all the states that admit a local hidden variable model (all separable states and some entangled ones, e.g.~\cite{werner1989,PhysRevA.73.062105,almeida2007,PhysRevLett.111.160402})
can be simulated with sufficient amount of shared randomness, bringing no gain to the quantum protocol.
However, if the size of the assisting resources is the same, states that do not violate any Bell inequality may possibly help improve the efficiency of the quantum protocol over the best classical ones.
We stress that this reasoning is not specific to random access codes, but also applies more generally to any task assisted with correlated resources.
This suggests that other correlation-assisted communication protocols can find a similar advantage using separable states.

We therefore restrict the amount of shared classical bits to be the same as the amount of shared quantum bits and study in detail the case of two assisting (qu)bits.
We first show how classical codes gain additional efficiencies by utilising the bias of the assisting bits to avoid wrong guesses.
Next, we provide quantum protocols assisted by separable states that outperform the best classical protocols, and show that in some cases they outperform even protocols assisted by quantum entanglement.

%%%%%%%%%%%%%%%%%%%%%%%%%%%%%%%%%%%%%%%%%%%%%%%%%%%%%%%%%%%%%%%%%%%%%%%

\section{Finite classical randomness}

A standard figure of merit characterising the efficiency of the RAC is the probability $P_{\min}$ of Bob's correct guess in the worst case scenario (minimised over $x$ and $i$), i.e. $P_{\min} = \min_{x,i} \mathrm{Pr}(b_i = x_i)$.
Let us recall that we are interested in the $n \to 1$ scenario with only one bit of classical communication from Alice to Bob and the input of Alice contains at least two bits, $n \ge 2$.
In addition to classical communication Alice and Bob may share assisting classical bits from a common source.
If no randomness is allowed in this scenario $P_{\min} = 0$, as there is always a bit that Bob guesses wrongly~\cite{ambainis2008}.
In the presence of shared randomness $r$, the efficiency $P_{\min}$ is additionally averaged over the assisting random bits, $P_{\min} = \min_{x,i} \sum_r p_r \mathrm{Pr}(b_i = x_i|r)$.
The following theorem characterises the maximal $P_{\min}$ in the presence of two bits of shared randomness.

\begin{theorem}
A classical $n \to 1$ RAC assisted with two bits from a common source has (i) $P_{\min} \le \frac{1}{2}$ if $n>2$; (ii) $P_{\min} \le \frac{2}{3}$ if $n=2$; (iii) $P_{\min} \le \frac{1}{2}$ for all $n > 1$ if the assisting bits have maximally mixed marginal for Bob.
\label{TH_MAX_MIXED}
\end{theorem}
\begin{proof}
(i) Let us denote the random bits of Alice and Bob by $r_a$ and $r_b$, respectively.
Alice's classical communication (encoding) is an output of a binary function $c = c(x,r_a)$, and Bob's guess of $x_i$ is also an output of a binary function $b_i = b_i(c,r_b)$.
Observe that for every given input $x$ Alice can only choose from the following four possible encoding functions:
1) $c = 0$ independently of $r_a$; 2) $c = 1$ independently of $r_a$; 3) $c = r_a$; 4) $c = 1 \oplus r_a$, where $\oplus$ denotes the binary sum.
Of course for different values of $x$ Alice can (and should) choose different encoding functions.
Indeed, we prove that if for two different inputs $x$ and $x'$ Alice uses the same encoding function, the probability of correct guess is no greater than $\frac{1}{2}$.
This is done via contradiction, for if Alice chooses for $x$ and $x'$ the same encoding function from the options above, her message $c$ is the same for both $x$ and $x'$.
Accordingly, since Bob is receiving the same message for both inputs, his guesses of the individual bits of $x$ and $x'$ for a given $r_b$ are the same and this implies that the probability of any given  guess is also the same (both for a fixed $r_b$ and averaged over the shared randomness).
Therefore, if his guesses are correct for the bits of $x$ with probability more than $\frac{1}{2}$, the guess of the differing bit of $x'$ must be incorrect with the same probability.
Hence, $P_{\min} \le \frac{1}{2}$.
Any sound strategy must therefore employ different encoding functions of Alice.
Since there are only four different such functions for a fixed input $x$, the efficiency is at most $\frac{1}{2}$ for all $n \ge 3$.
There is simply not enough shared randomness for Alice and Bob to do more.

(ii) We now focus on the $2 \to 1$ RAC.
In every protocol run, i.e. for a fixed $x$, Bob needs to prepare guesses $b_1$ and $b_2$ for the individual bits of Alice's input, which we order as $g_{c,r_b} = (b_1,b_2)$,
with indices $c, r_b$ describing the variables accessible to Bob.
Employing a method similar to Ref.~\cite{ambainis2002}, we define points $P(x) = (\mathrm{Pr}(b_1 = 1|x), \mathrm{Pr}(b_2 = 1|x))$ that represent the probabilities of  Bob's guesses being equal to $1$ for a fixed $x$.
Using the Bayes rule $\mathrm{Pr}(b_i = 1|x) = \sum_{r_a,r_b} p_{r_a r_b} \mathrm{Pr}(b_i = 1|r_a, r_b, x)$, and the fact that $r_a$, $r_b$, and $x$ deterministically specify $b_i$, i.e. $\mathrm{Pr}(b_i = 1|r_a, r_b, x) = b_i$, we
explicitly write the points corresponding to the four Alice's encoding functions listed above:
\begin{eqnarray}
    P_1(x^1) & = & p_{00} g_{0,0} + p_{01}  g_{0,1} + p_{10}  g_{0,0} + p_{11}  g_{0,1}, \label{EQ_Pa}\\
    P_2(x^2) & = &  p_{00}  g_{1,0} + p_{01}  g_{1,1} + p_{10}  g_{1,0} + p_{11} g_{1,1}, \label{EQ_Pb}\\
    P_3(x^3) & = & p_{00}  g_{0,0} + p_{01}  g_{0,1} + p_{10}  g_{1,0} + p_{11} g_{1,1}, \label{EQ_Pc}\\
    P_4(x^4) & = & p_{00}  g_{1,0} + p_{01}  g_{1,1} + p_{10}  g_{0,0} + p_{11} g_{0,1} , \label{EQ_Pd}
\end{eqnarray}
where $p_{kl} \equiv \mathrm{Pr}(r_a = k, r_b = l)$ is the distribution of the common source of randomness, $x^j$ denote the four different values of $x$ with index $j$ denoting the different encodings employed.
In order to achieve $P_{\min} > \frac{1}{2}$ all four guesses $g_{c,r_b}$ must be different for different values of $c$ and $r_b$.
Namely, in the decomposition of every point $P_j(x^j)$ above there must be a guess $g_{c,r_b} = x^j$, i.e. with the same individual bits as those of $x^j$.
If this is not the case then every guess $g_{c,r_b}$ contains at least one individual bit that is guessed wrongly.
Hence the probability of an individual bit being correct is equal to the probability of the other bit being a wrong guess, i.e. $P_{\min} \le \frac{1}{2}$.

We have shown that efficient codes must involve guesses $g_{c,r_b}$ with all different values for different $c$ and $r_b$.
We will now find the optimal strategy maximising $P_{\min}$ only for inputs $x^3$ and $x^4$ in Eqs.~(\ref{EQ_Pc}) and~(\ref{EQ_Pd}).
Since only two inputs are considered and the definition of $P_{\min}$ includes minimisation over all four inputs, this maximisation gives an upper bound on $P_{\min}$.
It will turn out that this upper bound is achieved.
In the best case, Bob never outputs a guess with both individual bits guessed wrongly.
Assume they are $g_{1,1}$ and $g_{0,1}$ in Eqs.~(\ref{EQ_Pc}) and~(\ref{EQ_Pd}), respectively.
Therefore, the best case corresponds to $p_{11} = 0$.
Since a guess giving the two bits of $x^3$ correctly must be different from the guess giving the two bits of $x^4$ correctly,
and the probability of guessing any individual bit is a sum of $p_{kl}$ corresponding to $g_{c,r_b} = x^j$ and $g_{c,r_b} = \bar x^j$ having the other individual bit flipped,
one may verify that
\begin{equation}
P_{\min} = \min(p_{00} + p_{01}, p_{00} + p_{10}, p_{01} + p_{10}).
\end{equation}
This is maximised for the biased distribution $p_{00} = p_{01} = p_{10} = \frac{1}{3}$, which implies that the optimal value is $P_{\min} = \frac{2}{3}$.
The optimal code, achieving $P_{\min} = \frac{2}{3}$, is detailed in Table~\ref{tab:proto} where Alice's encoding and Bob's output is completely specified.

\begin{table}
\begin{center}
\caption{\label{tab:proto} Details of the optimal $2 \to 1$ classical random access code assisted with two bits from a common source, $r_a$ and $r_b$.
These bits are distributed according to the biased distribution $p_{00} = p_{01} = p_{10} = \frac{1}{3}$ and $p_{11} = 0$, which is why below we do not present the case of $(r_a,r_b) = (1,1)$.
The guesses of Bob are denoted as $g_{c,r_{b}}$ and we note that for a given input $x$ they never contain both individual bits opposite to $x$.
In this sense the biased randomness is used to avoid giving wrong answers.
By comparing the individual bits of $g_{c,r_{b}}$ with the individual bits of $x$ one confirms that $P_{\min} = \frac{2}{3}$.}
\begin{tabular}{c c c c c c c c c}
\hline \hline
$x$ & & $(r_{a},r_{b})$ & & $c(x,r_{a})$ & & $g_{c,r_{b}}$ & & $P(x)$\tabularnewline
\hline

00 & & (0,0) & & 0 & & (0,1) & & $(\frac{1}{3},\frac{1}{3})$ \tabularnewline

 & & (0,1) & & 0 & & (0,0) & & \tabularnewline

 & & (1,0) & & 1 & & (1,0) & & \tabularnewline

% & & (1,1) & & 1 & & (1,1) & & \tabularnewline

\hline

01 & & (0,0) & & 0 & & (0,1) & & $(0,\frac{2}{3})$ \tabularnewline

 & & (0,1) & & 0 & & (0,0) & & \tabularnewline

 & & (1,0) & & 0 & & (0,1) & & \tabularnewline

% & & (1,1) & & 0 & & (0,0) & & \tabularnewline
\hline
10 & & (0,0) & & 1 & & (1,0) & & $(1,\frac{1}{3})$ \tabularnewline

 & & (0,1) & & 1 & & (1,1) & & \tabularnewline

 & & (1,0) & & 1 & & (1,0) & & \tabularnewline

% & & (1,1) & & 1 & & (1,1) & & \tabularnewline
\hline

11 & & (0,0) & & 1 & & (1,0) & & $(\frac{2}{3},\frac{2}{3})$ \tabularnewline

 & & (0,1) & & 1 & & (1,1) & & \tabularnewline

 & & (1,0) & & 0 & & (0,1) & & \tabularnewline

% & & (1,1) & & 0 & & (0,0) & & \tabularnewline
\hline \hline
\end{tabular}
\end{center}
\end{table}

(iii) Here we again utilize the fact that Bob's guesses $g_{c,r_b}$ must be different for different values of $c$ and $r_b$.
Since~(\ref{EQ_Pa}) involves the marginal distribution of Bob, the assumption of maximal mixedness gives $P_1(x^1) =\frac{1}{2} g_{0,0} + \frac{1}{2} g_{0,1}$, and
there is always an individual bit of $x^1$ that is guessed with probability $\frac{1}{2}$, thus $P_{\min} \le \frac{1}{2}$.
\end{proof}

We would like to emphasise that studies of randomness usually employ so-called ``common randomness'',
i.e. pairs of perfectly correlated and locally completely random bits,
whereas our proof shows that one can utilise the bias in the shared randomness to gain additional efficiency,
in this case to avoid giving wrong answers (see Table~\ref{tab:proto}).
%
%\begin{figure}
%\qquad \qquad \qquad \includegraphics[scale=0.4]{square.jpg}
%\caption{The points within the square represent the probabilities $(\mathrm{Pr}(b_1 = 1), \mathrm{Pr}(b_2 = 1))$ of Bob's guesses being equal to $1$.
%The quadrant $Q_{x}$ contains all the points giving rise to Bob's correct guess of Alice's individual inputs $x_1$ and $x_2$ being more than $\frac{1}{2}$ (excluding the lines bisecting the square).
%The points connected with dashed lines represent the best classical RAC assisted with two bits from a common source and show that at best $P_{\min} = \frac{2}{3}$.
%}
%\label{FIG_SQUARE}
%\end{figure}
%

%%%%%%%%%%%%%%%%%%%%%%%%%%%%%%%%%%%%%%%%%%%%%%%%%%%%%%%%%%%%%%%%%%%%%%%

\section{Finite quantum randomness}

Having established the classical bounds we proceed to demonstrate quantum protocols that exceed them.
We present explicit $2 \to 1$ and $3 \to 1$ quantum random access codes assisted with two correlated qubits.
These special cases are of particular interest because they may be concatenated to generate more general $n \rightarrow 1$ quantum codes (see Ref.~\cite{pawlowski2010} for a detailed discussion of this procedure).
After introducing the notation and essential concepts, we present detailed protocols and study their efficiency when assisted with Bell diagonal states.

Throughout the rest of the paper we employ the Bloch representation of qubit states and measurements, i.e. the three dimensional vector $\vec{s}$ represents the qubit state
$\rho(\vec{s})=\left(\hat 1+\vec{s}\cdot\vec{\sigma}\right)/2$, where $\vec{\sigma}\equiv (\sigma_1,\sigma_2,\sigma_3)$ is the vector of Pauli matrices $\sigma_x,\sigma_y, \sigma_z$.
A unit vector $\hat{\alpha}$ represents an ideal measurement with the probability of obtaining a measurement outcome $\alpha=0,1$, when measured on the state $\rho$, being $\mathrm{Tr}\left(\frac{1+(-1)^\alpha\hat{\alpha}\cdot\vec{\sigma}}{2}\rho\right)$.

A general two-qubit state is of the form $\rho_{ab}=\frac{1}{4}(\hat 1\otimes\hat 1+\vec{a}_0\cdot\vec{\sigma}\otimes\hat 1
+\hat 1\otimes\vec{b}_0\cdot\vec{\sigma}+\sum_{l,m=1}^{3}E_{ij}\sigma_l\otimes\sigma_m)$, where $\vec{a}_0$ and $\vec{b}_0$ are the local Bloch vectors of Alice and Bob, respectively. The matrix $E$ is the correlation matrix, and can always be made diagonal by an appropriate choice of local bases~\cite{horodecki1996b}. We therefore assume, without loss of generality, that the reference frames are appropriately chosen such that  $E=\mathrm{diag}(E_1,E_2,E_3)$. If $E_i\neq0$ we say that the state is correlated along that axis. We also make use of the fact that if Alice performs a measurement $\hat{\alpha}$ with outcome $\alpha$ on her half of the system, then Bob's post-measurement Bloch vector is:
\begin{equation} \label{eq:bstate}
    \vec{b}(\alpha)=\frac{\vec{b}_0+(-1)^{\alpha}E^T \hat{\alpha}}   {1+(-1)^{\alpha}\hat{\alpha}\cdot\vec{a}_0},
\end{equation}
where $E^T$ is the transposed matrix $E$ (not necessary if $E$ is already diagonal).

We shall explore the relationship between our protocols and a class of quantum correlations referred to as quantum discord~\cite{henderson2001,ollivier2001,celeri2011,modi2012}.
Specifically, we employ the normalized geometric measure of quantum discord \cite{dakic2010}.
A general zero-discord state has the form $\sigma_{ab}=p_0\rho_0\otimes|0\rangle\langle0|+p_1\rho_1\otimes|1\rangle\langle1|$, and the normalized geometric discord of $\rho_{ab}$ is defined to be~\cite{dakic2012}: $D_{a|b}^2(\rho_{ab})\equiv 2\mathrm{Min}_{\sigma}\mathrm{Tr}(\rho_{ab}-\sigma_{ab})^2$. For Bell diagonal states we have $D_{a|b}^2 = \frac{1}{2}(E_2^2 + E_3^2)$, where it is assumed that $E_1^2$ is the biggest among squared diagonal elements of $E$.

%%%%%%%%%%%%%%%%%%%%%%%%%%%%%%%%%%%%%%%%%%%%%%%%%%%%%%%%%%%%%%%%%%%%%%%

\subsection{$3\rightarrow1$ code}

The codes presented here are similar to the codes assisted with quantum entanglement~\cite{pawlowski2010},
with the key difference in the choice of Alice's measurements.
We focus first on the class of Bell diagonal states $\rho_{ab}$ correlated along all three axes $x$, $y$ and $z$:
\begin{equation} \label{eq:mmmargin}
    \rho_{ab}=\frac{1}{4}\left(\hat 1\otimes\hat 1+\sum_{i=1}^3E_{l}\sigma_l\otimes\sigma_l\right),
\end{equation}
though the presented protocols can give better than classical results for more general assisting states
(e.g. it will be easy to verify that $\vec a_0$ can be arbitrary).
The protocol is as follows:
\begin{enumerate}
 \item[(i)] For input $x$, Alice performs the measurement characterised by the Bloch vector $\hat {\alpha}(x) = \vec {\alpha}(x) / |\vec {\alpha}(x)|$, where
 $\vec {\alpha}(x) = (\frac{(-1)^{x_1}}{E_1},\frac{(-1)^{x_2}}{E_2},\frac{(-1)^{x_3}}{E_3})$,
    \item[(ii)] Alice sends her measurement outcome $c = \alpha$ to Bob,
    \item[(iii)] To guess the $i$th bit of Alice, Bob measures along $\sigma_i$, obtains the outcome $\beta_i$, and puts $\beta_i \oplus c$ as the guess.
\end{enumerate}
To grasp the mechanism of this protocol, note that depending on the input, $x$, Alice's measurement vectors point towards the vertices of a cuboid embedded in the Bloch sphere (see Fig.~(\ref{spheres}a)).
As a result of her measurement (with outcome $\alpha$) and the correlations in the shared state, the post-measurement local Bloch vector on Bob's side,
$\vec{b}(\alpha) = (-1)^\alpha ((-1)^{x_1},(-1)^{x_2},(-1)^{x_3})/|\vec \alpha|$, points towards one of the vertices of an inner cube within the Bloch sphere (see Fig.~(\ref{spheres}b)).
The $x$, $y$ and $z$ axes correspond to the direction of Bob's measurement, depending on whether he is guessing the first, second or third bit respectively.
Depending on her measurement outcome, Alice knows that Bob's post-measurement Bloch vector is either pointing towards the vertex of the cube encoding $x$ or the vertex directly opposite across the origin,
encoding $\bar x$ with all individual bits flipped.
Therefore, Alice sends a message to Bob to either flip his guess or not to flip it.
Note that the inner product of Bob's measurement vectors (along the axes) with a vector pointing to any vertex is the same, up to a sign.
The probability of correct guess of every individual bit is therefore the same, giving $P_{\min}$, and for Bell diagonal states it is equal to
\begin{equation}
P_{\min} = \frac{1}{2}\left(1 + \frac{1}{\sqrt{E_1^{-2} + E_2^{-2} + E_3^{-2}}} \right).
\label{3TO1}
\end{equation}
Since $P_{\min} > \frac{1}{2}$, this quantum code is thus more efficient than the best classical code (see Th.~\ref{TH_MAX_MIXED} and note that Bell diagonal states have maximally mixed marginals).

\begin{figure}
  \centering
  \includegraphics[scale=0.4]{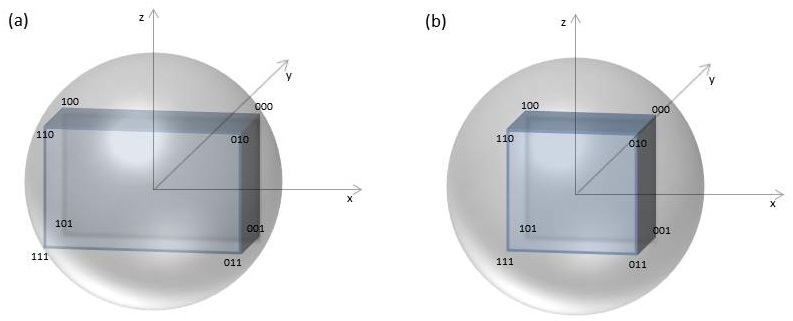}
  \caption{The quantum $3\rightarrow 1$ random access code assisted with two qubits. (a) Alice's measurement vectors point towards the vertices of a cuboid embedded within the Bloch sphere.
  The cuboid is defined by the correlations of the shared quantum state.
  (b) Bob's post-measurement Bloch vectors point towards the vertices of an inner cube centered at the origin. See main text for explanation how the code works.}\label{spheres}
\end{figure}

%%%%%%%%%%%%%%%%%%%%%%%%%%%%%%%%%%%%%%%%%%%%%%%%%%%%%%%%%%%%%%%%%%%%%%%

\subsection{$2\rightarrow1$ code}

This code can operate on a slightly broader class of states as we now allow $E_3$ to vanish.
The protocol follows the same procedures as in the $3\rightarrow 1$ case, with the exception that Alice's measurements are given by $\vec {\alpha}(x) = (\frac{(-1)^{x_1}}{E_1},\frac{(-1)^{x_2}}{E_2},0)$,
the efficiency of this quantum code can then be verified to be
\begin{equation}
P_{\min} = \frac{1}{2}\left(1 + \frac{1}{\sqrt{E_1^{-2} + E_2^{-2}}} \right),
\label{2TO1}
\end{equation}
which is again better than the best classical protocol using bits with maximally mixed marginals (see Th.~\ref{TH_MAX_MIXED}).
Even if Alice and Bob were allowed to share more generally correlated classical bits, for which $P_{\min}$ may be as high as $\frac{2}{3}$ for the $2 \rightarrow 1$ case, the above code may nonetheless still outperform the best possible classical RACs so long as the assisting qubits are sufficiently strongly correlated.
It turns out that entanglement is not a necessary prerequisite to present such a quantum advantage.
We demonstrate this with concrete examples in the following section.

%%%%%%%%%%%%%%%%%%%%%%%%%%%%%%%%%%%%%%%%%%%%%%%%%%%%%%%%%%%%%%%%%%%%%%%

\section{Examples}

Consider Werner states, belonging to the class of Bell diagonal states and given by the mixture of white noise and maximally entangled state~\cite{werner1989}:
\begin{equation}
        \rho_{ab}=(1-q)\frac{\hat 1\otimes\hat 1}{4}+q|\psi\rangle\langle\psi|, \quad q\in[0,1].
\end{equation}
The state is entangled for $q>\frac{1}{3}$ and is separable otherwise.
Its geometric discord can easily be verified to be $D_{a|b}= q$ \cite{dakic2010}.
Since for the Werner states all $E_i = \pm q$, Eqs.~(\ref{3TO1}) and (\ref{2TO1}) reveal that the geometric discord directly measures the efficiency of the $n \to 1$ quantum codes
assisted with this class of states, $P_{\min} = \frac{1}{2}(1+ \frac{D_{a|b}}{\sqrt{n}})$ for $n=2,3$.
Moreover, it is the presence of quantum discord in the assisting states that empowers the quantum advantage.

The same statement likely holds for more general codes.
For example, concatenating $2 \to 1$ code assisted by the Werner state as in Ref.~\cite{pawlowski2010} one finds that the efficiency of $2^m \to 1$ code is given by
\begin{equation}
P_{\min} = \frac{1}{2} \left(1 +\left( \frac{D_{a|b}}{\sqrt{2}} \right)^{m} \right).
\label{PMIN_WERNER}
\end{equation}
The concatenation of the quantum codes requires $2^m - 1$ pairs of qubits in the Werner state and a fair comparison with the classical case is then made by replacing the qubit pairs with correlated bits that have maximally mixed marginals.
Numerical simulations indicate that $4\rightarrow 1$ classical RACs formed through the concatenation procedure cannot achieve $P_{\min} > \frac{1}{2}$.
We conjecture in general that the concatenation of $2 \rightarrow 1$ classical RACs assisted with bits having maximally mixed marginals cannot give $P_{\min} > \frac{1}{2}$,
and therefore that the quantum advantage is present for any $m$, as indicated in Eq.~(\ref{PMIN_WERNER}).

We now show that a separable state may be used to outperform the best classical code assisted with two correlated random bits.
The example once again utilises Bell diagonal states.
Recall that the classical bound is $P_{\mathrm{cl}}^{2 \to 1} = \frac{2}{3} \approx 0.667$ for all classical $2 \to 1$ RACs, and $P_{\mathrm{cl}}^{3 \to 1} = \frac{1}{2}$ for all classical $3 \to 1$ RACs.
By optimising the efficiency of the $2 \to 1$ quantum code, see Eq.~(\ref{2TO1}), over the separable Bell diagonal states one finds that the optimal state has $E_1=E_2=\frac{1}{2}$ and $E_3=0$, which gives the efficiency $P_{\min} = \frac{1}{2}(1+\frac{1}{2\sqrt{2}}) \approx 0.677$, slightly above the classical bound.
Better results are obtained for the $3 \to 1$ quantum code.
By optimising Eq.~(\ref{3TO1}) over separable Bell diagonal states, the best state has $E_1=E_2=E_3=\frac{1}{3}$ and the efficiency is $P_{\min} = \frac{1}{2}(1+\frac{1}{3\sqrt{3}}) \approx 0.596$, considerably above the classical bound.
Note that there may exist a quantum code achieving better efficiencies, utilising some other class of separable states or following a different procedure.

In the last example we show that separable states can outperform some entangled states.
We have already demonstrated that using a separable state the $2 \to 1$ quantum code may achieve efficiencies of at least $P_{\min} = \frac{1}{2}(1+\frac{1}{2\sqrt{2}})$.
Comparing this with Eq.~(\ref{PMIN_WERNER}), one can see that it outperforms the protocol assisted with the entangled Werner states for $\frac{1}{3} < q < \frac{1}{2}$.
It remains to show that there is no better quantum protocol for $2 \to 1$ quantum code assisted with the Werner states.
This follows from the optimality of the protocol for the maximally entangled state $|\psi\rangle$ shown in Refs.~\cite{pawlowski2010,ambainis2002},
the fact that the completely mixed state encodes local randomness giving at most $P_{\min} = \frac{1}{2}$, and that the Werner state is a mixture of these two states.

Although quantum discord empowers quantum advantage in our examples and is proportional to the efficiency of the protocol for fixed classes of states
it should be noted that the amount of the geometric quantum discord for different classes of states is not an indicator of the usefulness of the states for quantum random access codes.
Namely, our optimal separable state for $2 \to 1$ code has discord $D_{\mathrm{sep}} = \frac{1}{2 \sqrt{2}} \approx 0.354$ that corresponds to $P_{\min} = \frac{1}{2}(1 + D_{\mathrm{sep}})$,
but Werner states have $D_{\mathrm{Wer}} = q$ and the corresponding $P_{\min} = \frac{1}{2}(1 + \frac{D_{\mathrm{Wer}}}{\sqrt{2}})$.
Therefore, Werner states containing more discord than the separable state, i.e. $D_{\mathrm{Wer}} \in (\frac{1}{2 \sqrt{2}}, \frac{1}{2})$,  still give worse $P_{\min}$ than the separable state.
The precise physical quantity that is a resource for better quantum codes is at present unknown.

Finally we would like to comment on a variation of quantum random access codes that allows Alice to send a qubit to Bob in place of both of them sharing correlated qubits.
Results presented here may suggest that Alice should be able to send Bob noisy states (as opposed to pure states) and still be able to beat the classical limit.
This is indeed the case as we will briefly explain for the $2 \to 1$ code.
The classical limit is known to be $\frac{1}{2}$ if we allow for uncorrelated local randomness~\cite{ambainis2008}.
The optimal quantum protocol that beats this bound encodes the input $x$ into pure quantum states $\ket{\psi_x}$ with Bloch vectors $\vec \psi_x = ((-1)^{x_1}, (-1)^{x_2}, 0)$.
If Bob now measures along $x$ ($y$) axis in order to read the first (second) bit his worst case probability of correct guess is $\frac{1}{2}(1 + \frac{1}{\sqrt{2}})$.
Suppose we perform the same measurements on white noise $\frac{1}{2} \hat 1$. The outcomes of the measurements are completely random, which is enough to give as good a result as the best classical protocol in the worst case.
If Alice's encoding is in the form of a mixed state $q \proj{\psi_x} + (1-q) \frac{1}{2} \hat 1$,
its corresponding Bloch vector is $q \vec \psi_x$. Applying the same protocol, we find that $P_{\min} = \frac{1}{2}(1 + \frac{q}{\sqrt{2}})$,
which is always better than classical except for the completely mixed state of $q=0$.

%%%%%%%%%%%%%%%%%%%%%%%%%%%%%%%%%%%%%%%%%%%%%%%%%%%%%%%%%%%%%%%%%%%%%%%

\section{Conclusions}

We demonstrated that separable states are a useful resource in random access codes as soon as finite shared randomness in the quantum and classical protocols is counted in the same way, i.e. bits are replaced with qubits.
This is in particular relevant if randomness is not a freely available resource.

%One may also consider an extension of our protocol by replacing the classical communication channel with a quantum communication channel so that Alice and Bob can communicate qubits. The communication of a classical bit is therefore equivalent to passing a qubit through a dephasing channel. It is interesting to ask if such a fully quantum protocol is substantially more powerful than the ones we have discussed, but that lies outside the scope of this paper. In any case, our results indicate that in such a scenario, quantum enhancements may survive a noisy quantum communication channel so long as Alice and Bob share quantum correlations.

We hope the example given here opens a research avenue on efficiency of solutions to various problems in the presence of finite randomness.
This is of both practical and fundamental interest.
On the practical side, computers can use only a finite and restricted set of random bits for computations
and therefore separable states are likely to enlarge the class of states that allows quantum advantages once these restrictions are taken into account.
On the fundamental side, it would be interesting to know if entanglement is necessary to demonstrate in a Bell-like scenario deviations from predictions of local hidden variable models that involve only a finite number of bits.

\ack

We thank Kavan Modi for discussions.
This work is supported by the National Research Foundation, the Ministry of Education of Singapore grant no RG98/13, and start-up grant of the Nanyang Technological University.

\section{References}

\bibliographystyle{unsrt}
\bibliography{references}

\end{document}